\documentclass[pdflatex,sn-vancouver,Numbered]{sn-jnl}% Vancouver 
\usepackage{graphicx}%
\usepackage{multirow}%
\usepackage{amsmath,amssymb,amsfonts}%
\usepackage{amsthm}%
\usepackage{mathrsfs}%
\usepackage[title]{appendix}%
\usepackage{xcolor}%
\usepackage{textcomp}%
\usepackage{manyfoot}%
\usepackage{booktabs}%
\usepackage{algorithm}%
\usepackage{algorithmicx}%
\usepackage{algpseudocode}%
\usepackage{listings}%
\usepackage{enumerate}%
\usepackage{cases}%
\usepackage{xcolor}%

\newtheorem{theorem}{Theorem}
\newtheorem{cond}{Conditions}[subsection]

\theoremstyle{definition}%

\newcommand{\norm}[1]{\left\lVert#1\right\rVert}%

\begin{document}

\title[Article Title]{Exponential Separation Criteria for Quantum Iterative Power Algorithms}

\author*[1]{\fnm{Andr\'as} \sur{Cz\'egel}}

\author[1]{\fnm{Bogl\'arka} \sur{G.-T\'oth}}

\affil*[1]{\orgdiv{Department of Computational Optimization}, \orgname{University of Szeged}, \orgaddress{\street{\'Arp\'ad t\'er 2.}, \city{Szeged}, \postcode{6720}, \country{Hungary}}}

\abstract{In the vast field of Quantum Optimization, Quantum Iterative Power Algorithms (QIPA) has been introduced recently with a promise of exponential speedup over an already established and well-known method, the variational Quantum Imaginary Time Evolution (varQITE) algorithm. Since the convergence and error of varQITE are known, the promise of QIPA also implied certain collapses in the complexity hierarchy — such as NP $\subseteq$ BQP, as we show in our study. However the original article of QIPA explicitly states the algorithm does not cause any collapses. In this study we prove that these collapses indeed do not occur, and with that, prove that the promised exponential separation is practically unachievable. We do so by introducing criteria for the exponential separation between QIPA$_2$ that uses a double exponential function and varQITE, and then showing how these criteria require certain properties in problem instances. After that we introduce a preprocessing step that enforces problems to satisfy these criteria, and then we show that the algorithmic error blows up exponentially for these instances, as there is an inverse polynomial term between speedup and the error. Despite the theoretical results, we also show that practically relevant polynomial enhancement is still possible, and show experimental results on a small problem instance, where we used our preprocessing step to achieve the improvement.}

\maketitle

\section{Introduction}

Quantum optimization is at the forefront of quantum computational and information-theoretical research \cite{Abbas2023, Dalzell2023QuantumAA, Barends2015, QIPU24}. Most efforts focus on discrete optimization problems, either as global optimization problems or on more specific problems like linear or quadratic combinatorial optimization \cite{Quintero2022, DeSantis2024, Gacon2021}. Enhancements also appear, such as decomposition methods \cite{Zhao2022, Franco2023} or quantum resource optimization \cite{Franco2023, Chatterjee2024}.

Quantum Time Evolution (QTE) based approaches are also able to tackle these problems. In general, they describe the process of evolving a quantum state over time under some Hamiltonian $H$ \cite{Barends2015, Barison2021}. While real time evolution can address problems in the class BQP (Bounded-error Quantum Polynomial Time) \cite{Cubitt2018}, which is the quantum analogue of P, 
Quantum Imaginary Time Evolution (QITE) can address a wider class, which we believe contains even QMA-hard problems \cite{Kempe2004}, where QMA is the quantum complexity class Quantum Merlin-Arthur, which is analogous to the classical class NP. 

To utilize this on a gate-based quantum computer, translation of the time evolution into quantum gates is needed. The translation can be a mathematical approximation, e.g.\ Trotterization \cite{Motta2019, Lloyd1996}, however those circuits are usually too deep for the current quantum hardware. An alternative is to use shallow parameterized circuits. The variational QITE (varQITE) algorithm \cite{McArdle2019, Motta2019} in its various formulations uses this approach to simulate the time evolution.

Since combinatorial optimization problems usually involve an exponentially large number of possible configurations in the search space, algorithmic solutions often look for heuristic or numerical cuts, or alternatively an information-theoretic compression of certain regions, or even the entire search space. The Iterative Power Algorithm (IPA) \cite{Soley2021} utilizes the latter concept by using tensor trains \cite{Oseledets2011}, also formulated and referred as Matrix Product States (MPS \cite{stlund1995}), to efficiently compress and encode the full search space of certain kinds of problems. This is a fully classical approach, however it inherits the limitations of MPS, namely their compression capabilities hold only for well-structured problems. To enhance the algorithm, \cite{Kyaw2023} used quantum states and their variational approximation with the concept of IPA to overcome the weakness of Tensor Train representation. They introduced the Quantum Iterative Power Algorithms (QIPA). Since QIPA is a full quantum generalization of IPA, they investigated thoroughly two versions: the one with a simple exponential oracle, which turns out to be an equivalent formulation to the varQITE algorithm \cite{McArdle2019}. The other one uses a double exponential function, which we call QIPA${_2}$.

\subsection{Motivation}

In \cite{Kyaw2023}, we find a complexity analysis that compares varQITE to QIPA$_2$. Based on that analysis, the authors of the paper state that under some conditions, an exponentially faster convergence can be achieved. Besides that, however, they state that this does not introduce any collapse in complexity hierarchical relations. 

The two statements together grabbed our attention. How is that possible that an algorithm, given some conditions on problem structure, yield an exponential speedup, yet it does not affect problem class relations? Can we find a set of problems that we can now solve exponentially faster? What class do they belong? Or, alternatively, is it possible for both statements to hold at the same time? What if not? Our work focuses on these statements, by proving or disproving them, and the consequences that the answers to the questions entail.

\subsection{Methodology}

In the first part, we define conditions for a problem structure that yield the promised exponential separation in the algorithmic iterations.

In the second part, we show a small algorithmic addition that enables exponential separation for problems that naturally do not possess the structure to satisfy the conditions defined in the previous step.

After that we consider algorithmic error bounds that have not been discussed previously. From that, we prove the original statement of the article \cite{Kyaw2023} which states the QIPA$_2$ algorithm does not collapse complexity theoretical hierarchies, however we also show this entails that the promised theoretical exponential separation between the algorithm is practically inaccessible.

\section{Results}
\subsection{Exponential separation criteria}

First, we would like to introduce the convergence rate analysis based on the algorithmic bounds provided by the original studies \cite{Kyaw2023,Soley2021}.

From \cite{Kyaw2023} we can define a lower bound on the number of IPA iterations to amplify the amplitude of the solution to measure it with chance greater than $1/2$, for both varQITE and QIPA$_2$. Let us use the notation and assumptions of \cite{Soley2021}, which is the following. The unique solution is encoded by the largest eigenvalue of $H$ as $\lambda_1$ and any other eigenvalue could be noted by $\lambda_2$, where $\lambda_1 > \lambda_2 > 0$. Note that the original article \cite{Soley2021} considers all \textit{other} eigenvalues to have the same value. Furthermore, let $n$ denote the number of variables in the system and $\kappa_{\text{varQITE}}, \kappa_{\text{QIPA}_2}$ are lower bounds on the number of iterations for the respective algorithms. The bounds from \cite{Soley2021,Kyaw2023} are as follows:

% [todo: lambda, n integer]

\begin{equation}
    \kappa_{\text{varQITE}} = \frac{n}{\log(\lambda_1 / \lambda_2)},
\end{equation}
and

\begin{equation}
    \kappa_{\text{QIPA}_2} = \frac{n}{\lambda_1 - \lambda_2}.
    \vspace{16pt}
\end{equation}

Now consider the fact that we are looking for exponential separation, i.e. $\kappa_{\text{varQITE}} = \Omega \left(exp(n)\right)$ while $\kappa_{\text{QIPA}} = \mathcal{O} \left(poly(n)\right)$. We do that with the following methodology: to fulfill the exponential separation, we present boundary conditions for $\lambda_1$ and $\lambda_2$ in the number of variables $n$.

\subsubsection{Boundary conditions}
We formulate a system of inequalities with assumptions for an exponential separation. The assumptions on the number of steps can be written as follows.
\begin{equation}
    \begin{cases}
        \qquad \frac{n}{\log{\left(\lambda_1 / \lambda_2\right)}} \quad&\geq\quad c2^n  \\[6pt]
        \qquad\frac{n}{\lambda_1 - \lambda_2} \quad&\leq\quad dn^k \label{eq:ineq2} \\[6pt]
        \qquad \qquad \lambda_1 \quad&>\quad \lambda_2
    \end{cases} \\
\end{equation}
where $c, d, \text{and }k$ are problem specific positive real constants.

From the inequalities above, we obtain the following conditions on the exponential separation and of $\lambda_1$, $\lambda_2$ and $n$. For detailed derivations and proof for the conditions, see Appendix \ref{app:a}.

\begin{cond} ~
    \begin{enumerate}[I.]
\label{conditions}
    \item $\qquad\lambda_1 - \lambda_2 \geq 1/\text{poly}(n) \sim 1/\text{poly}(\log(\lambda_2))$ 
    \item $\qquad\lambda_1/\lambda_2 \sim 1+1/\exp(n)$ 
    \item $\qquad\lambda_1, \lambda_2 \notin \mathcal{O}(\text{poly}(n))$
\end{enumerate}
\end{cond}

In the IPA article \cite{Soley2021}, the authors used an explicit mapping between a discrete minimization problem and a discrete maximization problem. If we omit the transformation, we arrive at a classical discrete optimization problem's quantum representation. For this problem, we obtained the following structural conditions:

\begin{enumerate}
    \item The optimal value scales super-polynomially with the input size.
    \item The optimal value is inverse super-polynomially close to the closest feasible solution value.
\end{enumerate}

In this scenario, the varQITE needs exponentially many iterations to reach an error less than $1/2$. At the same time, QIPA with the double exponential function needs only poly($n$) iterations.

\subsection{Artificial separation enforcement}

We have separation statements for the eigenvalues of $H$ both in relative and absolute scale. If both gaps are small, we do not have exponential separation between varQITE and QIPA$_2$. In this case, we can \textit{artificially} introduce the necessary gaps by \textit{upscaling} the Hamiltonian by a factor $\alpha$.

\begin{equation}
    \Hat{H} = \alpha H
\end{equation}

This operation does not affect the corresponding eigenvectors, however it scales the eigenvalues by $\alpha$. That implies for the eigenvalues of $\hat{H}$, denoted by $\hat{\lambda_1}$ and $\hat{\lambda_2}$,
\begin{align}
    \frac{\hat\lambda_1}{\hat\lambda_2} ~&=~ \frac{\lambda_1}{\lambda_2}, \\[8pt]
    \hat\lambda_1 - \hat\lambda_2 ~&=~ \alpha \left(\lambda_1 - \lambda_2\right).
\end{align}

This upscaling of the problem Hamiltonian is able to put the problem under the criteria above, i.e.\ Conditions \ref{conditions}.
However, when the absolute gap $\left(\lambda_1-\lambda_2\right) \sim 1/\exp(n)$, for exponential separation, $\alpha$ must be of order $\exp(n)$. We also show a detailed example in Appendix \ref{app:example}.

\subsection{Error bounds}
\label{sec:error}

With the artificial enforcement, and taking both single step time complexity and overall approximation complexity into account, it is theoretically possible to achieve the exponential separation for any problem with sufficiently large upscaling factor.

That is, one QIPA iteration requires polynomially many operations in the problem size, and we need polynomially many iterations to boost the solution amplitude as well. This concept would describe a Quantum-Classical Fully Polynomial-Time Approximation Scheme (QCFPTAS) for the desired measurement probability of the solution, that is strictly larger than 1/2. Then, our measurement probability is easy to drive above 2/3, which effectively would mean this algorithm places NP-complete problems like the Maximum Cut (see Appendix \ref{app:example}) into BQP. That would also mean NP $\subseteq$ BQP which is widely conjectured not to be true, however without considering the implicit algorithmic approximation errors of the algorithms themselves, it holds.

As pointed out above, we have not considered new algorithmic error bounds coming from our criteria, nor implicit approximation subroutines errors. We show that the statement in \cite{Kyaw2023}, which states that the QIPA does not affect complexity hierarchical relations, holds. However, we also show that the exponential separation is practically unachievable. 

Since the upper bound of the error in QIPA$_2$ is higher than in varQITE as pointed out in \cite{Kyaw2023},

\begin{equation}
    \varepsilon_{\text{\tiny{QIPA}}_2} \geq \varepsilon_{\text{\tiny{varQITE}}} + \Delta\delta\tau + \mathcal{O}(\delta\tau^{3/2}) ,\label{eq:error}
\end{equation}
where
\begin{equation}
    \Delta^2 = \left\langle \left( \left(1+e^{H\delta t}\right)^2 / \left(\delta \tau^2\right) + 2 \left(e^{H\delta t}-1\right)H/(\delta\tau) - \left(e^{H \delta t}-2\right)H^2 \right) \right\rangle
\end{equation}
for $\delta$ being the step size, $\tau$ the time of the imaginary evolution and $H$ is the time evolution Hamiltonian.

The extra error term $\Delta\delta\tau$ depends linearly on $H$, as $\Delta^2$ depends on the expectation value of a sum of operators in which the dominating term is quadratic in H. By our Conditions \ref{conditions}, we know that the spectrum of $H$ scales super-polynomially, and together with this error term, they are sufficient to prove that the exponential separation yields a super-polynomial blow up in the error of the algorithm.

However, we can give a stronger result examining the error of base algorithm, varQITE, under our generalization by upscaling. Since upscaling can put any problem from above into the exponential separation regime, it is also relevant that varQITE's own error also blows up exponentially upon exponential upscaling (see Appendix \ref{app:varQITE}), because its error depends on the variance of $\hat{H}$.

\section{Discussion}

Before further discussion, we would like to highlight the fact that as we skipped over the algorithmic error terms before, we also skipped another, practically important factor: numerical calculations on a noisy quantum device. Both the basic varQITE and the discussed QIPA use some kind of variational principle. If we take the common choice, McLachlan's variational principle \cite{McLachlan1964} as an example, we find a linear system of equations for parameter change calculations. The system's solution method is purely classical, however the preparation of the matrix and the constant vector (see $F$ and $C$ in \eqref{eq:step_error}-\eqref{eq:C}) is a quantum subroutine for each entry of the respective elements.

This is not negligible, and in practice, calculating those two elements provides a major component of the error of the algorithms we discussed. The error calculation for practical usefulness is outside the scope of this study, however we would like to state here that QITE already has a dual formulation \cite{Gacon2024}, which does not need to calculate those two elements at all. On the other hand, for practical error considerations in intermediate-scale quantum algorithms, \cite{Bharti2022} provides a great perspective.

Despite the fact that exponential separation is theoretically unreachable, a polynomial one is possible. Since this is a wide statement and includes most practical questions, this could be an interesting and relevant question in a future study.

Another interesting approach could be to use Neural Quantum States (NQS) \cite{Lange2024} in QIPA with the aim of exploring the whole search space of a discrete problem, while taking advantage of the noise-free efficient approximate compression of the quantum system that NQS provide.

\section{Conclusion}\label{sec13}

In our study, we derived how the Quantum Iterative Power Algorithms do not yield exponential speedup, nor collapses in complexity theoretical hierarchies. We proved the non-collapse statement of \cite{Kyaw2023} and provided clarification of its promises. Furthermore we showed that polynomial enhancement is possible, supporting it with an example and experimental results. Our main conclusion is that this algorithm family theoretically does not imply exponential separation, and on the other, practically relevant side, a more in-depth study could show practical improvements over varQITE and possibly other algorithms.

\backmatter

\begin{appendices}

\section{Numerical conditions}
\label{app:a}

We derive the numerical conditions from \eqref{eq:ineq2}.

Starting with the first inequality, 
\begin{align}
    c2^n \quad&\leq\quad \frac{n}{\log{\left(\lambda_1 / \lambda_2\right)}} \\[6pt]
    c2^n \log{\left(\lambda_1 / \lambda_2\right)} \quad&\leq\quad n\label{eq:lambda1leq1} \\[6pt]
    \log{\left(\lambda_1 / \lambda_2\right)} \quad&\leq\quad \frac{n}{c2^n} \label{eq:lambda1leq2} \\[6pt]
    \lambda_1 / \lambda_2 \quad&\leq\quad 2^{\frac{n}{c2^n}} \label{eq:lambda1leq3}\\[6pt]
    \lambda_1 \quad&\leq\quad 2^{\frac{n}{c2^n}}\lambda_2\label{eq:lambda1leq4}
\end{align}
where \eqref{eq:lambda1leq1} is allowed since $\lambda_1 > \lambda_2$, that means $\left(\lambda_1 / \lambda_2\right) > 1$ thus $\log{\left(\lambda_1 / \lambda_2\right)} > 0$. 
The second step \eqref{eq:lambda1leq2} is allowed because $n > 0$ and $c > 0$, thus $c2^n > 0$. The next step to \eqref{eq:lambda1leq3} is valid because the log is of base 2, monotonic, and increasing function.
The last step \eqref{eq:lambda1leq4} is due to $\lambda_2>0$.
% ezt hogy kellene szepen?

Moving to the second inequality of \eqref{eq:ineq2},
\begin{align}
        \frac{n}{\lambda_1 - \lambda_2} \quad&\leq\quad dn^k \\[6pt]
        n \quad&\leq\quad dn^k\left(\lambda_1 - \lambda_2\right) \label{eq:poly1} \\[6pt]
        \frac{n}{dn^k} \quad&\leq\quad  \lambda_1 - \lambda_2 \label{eq:poly2} \\[6pt]
        \frac1{dn^{k-1}} \quad&\leq\quad  \lambda_1 - \lambda_2 \label{eq:poly4}
\end{align}
where \eqref{eq:poly1} is allowed because $\lambda_1 > \lambda_2$ thus $\lambda_1 - \lambda_2 > 0$. \eqref{eq:poly2} is because $d>0$ and $n^k>0$, and so is their product.

By substituting the first result, \eqref{eq:lambda1leq4} into the second, \eqref{eq:poly4},
\begin{align}
        \frac1{dn^{k-1}} \quad&\leq\quad  \lambda_1 - \lambda_2 \\[6pt]
        \frac1{dn^{k-1}} \quad&\leq\quad  2^{\frac{n}{c2^n}}\lambda_2 - \lambda_2 \label{eq:conn1}\\[6pt]
        \frac1{dn^{k-1}} \quad&\leq\quad  \lambda_2\left(2^{\frac{n}{c2^n}} - 1\right) \\[6pt]
        \frac1{dn^{k-1} \left(2^{\frac{n}{c2^n}} - 1\right)} \quad&\leq\quad  \lambda_2 \label{eq:conn3}
\end{align}
First step \eqref{eq:conn1} is an over-estimation using \eqref{eq:lambda1leq4}, however since we're looking for the minimum possible value of $\lambda_1$ and $\lambda_2$, and since their growth is monotonic in $n$, we're interested in the case where the equality holds. The third step \eqref{eq:conn3} is legal since $n>0$ and $c>0$ and thus $2^{\frac{n}{c2^n}} - 1 > 0$.

We can also derive a similar lower bound for $\lambda_1$. For that, we express $\lambda_2$ from \eqref{eq:lambda1leq4} as
\begin{align}
    \frac1{2^{\frac{n}{c2^n}}} \lambda_1 \quad&\leq\quad \lambda_2. \label{eq:l1l2relative2}
\end{align}

Same as in the previous step: legal since $n>0$ and $c>0$.

\begin{align}
        \frac1{dn^{k-1}} \quad&\leq\quad  \lambda_1 - \lambda_2 \\[6pt]
        \frac1{dn^{k-1}} \quad&\leq\quad  \lambda_1 - {2^{\frac{-n}{c2^n}}} \lambda_1 \\[6pt]
        \frac1{dn^{k-1}} \quad&\leq\quad  \lambda_1 \left(1 - {2^{\frac{-n}{c2^n}}}\right) \\[6pt]
        \frac1{dn^{k-1} \left(1 - {2^{\frac{-n}{c2^n}}}\right)} \quad&\leq\quad  \lambda_1  
\end{align}

Now, we express the order of $\lambda_2$ as a function of $n$, namely, that $\lambda_2$ needs to be super-polynomial in $n$.
As a consequence, $\lambda_1$ also have to be super-polynomial in $n$.

\begin{theorem}
\label{prop:lambaexp}
For the exponential separation to hold between varQITE and QIPA$_2$, the problem Hamiltonian's second largest eigenvalue $\lambda_2$, must be super-polynomial in $n$, formally, for any problem and any positive real $\ell$ 
\begin{equation}
    \lambda_2 \quad\neq\quad \mathcal{O}(n^\ell).
\end{equation}

\end{theorem}

\begin{proof}

Let us reformulate the statement mathematically by using our lower bound on $\lambda_2$ from \eqref{eq:conn3}:

\begin{equation}
    \frac1{dn^{k-1} \left(2^{\frac{n}{c2^n}} - 1\right)} \quad\neq\quad \mathcal{O}(n^\ell)
\end{equation}
for a real positive constant $\ell$.

 First, let us relax $n$ to a non-negative real, and recall the condition that for any function $f,g: \mathbb{R}^+\to \mathbb{R}^+$
\begin{equation*}
    f(n) = \mathcal{O}(g(n)) \quad \Longleftrightarrow \quad \lim_{n \to \infty} \frac{f(n)}{g(n)} = \beta \qquad 0\leq\beta \in \mathbb{R}.
\end{equation*}
Thus, if $\beta=\infty$, $f(n) \neq \mathcal{O}(g(n))$.

\vspace{16pt}
\noindent
Let us show it indirectly. Suppose that
\begin{equation}
    \frac1{dn^{k-1} \left(2^{\frac{n}{c2^n}} - 1\right)} \quad=\quad \mathcal{O}(n^\ell) \label{eq:setup}
\end{equation}%
thus 
\begin{align}
        \beta &= \lim_{n \to \infty} \frac1d \frac{1}{n^\ell n^{k-1} \left(2^{\frac{n}{c2^n}} - 1\right)} \\[6pt]
        &=\frac1d \lim_{n \to \infty} \frac1{n^{\ell+k-1} \left(2^{\frac{n}{c2^n}} - 1\right)} \label{eq:lim1} 
\end{align}

From limit theoretical principles \cite{limits}, we recall that for any $f$ on (0,$\infty$), $\lim_{x\to\infty}f(x)$ exists if and only if $\lim_{x \to 0+}f(\frac1x)$ exists and in this case, the two limits are equal. Combining that with the fact that $\lim_{x \to 0} \frac{2^x - 1}{x} = \ln{2}$, we get 
\begin{equation}
    \lim_{x \to \infty} \frac{2^\frac1x - 1}{\frac1x} = \ln{2}.
\end{equation}

Now let us substitute $x$ with the term $\frac{c2^n}{n}$. It is continuous and monotonically increasing for $n > \frac1{\ln2}$, that assures $\frac{c2^n}{n} \to\infty$ as $n\to\infty$. Thus we arrive at
\begin{equation}
    \lim_{n \to \infty} \frac{2^{\frac{n}{c2^n}} - 1}{\frac{n}{c2^n}} = \ln{2}.
\end{equation}
This can be reformulated into an expression that we can substitute straight into \eqref{eq:lim1}. We use the fact that  $\frac1{\lim_{x \to a} f(x)} = \lim_{x \to a} \frac1{f(x)}$, so
\begin{equation}
    \lim_{n \to \infty} \frac{\frac{n}{c2^n}}{2^{\frac{n}{c2^n}} - 1} = \frac1{\ln{2}}. \label{eq:lim_sub}
\end{equation}

\vspace{20pt}
\noindent
Now we rearrange \eqref{eq:lim1} by multiplying it with $1 = \frac{c2^n}{n} \cdot \frac{n}{c2^n}$, and apply \eqref{eq:lim_sub} as follows
\begin{align}
        \beta &= \frac1d \lim_{n \to \infty} \frac{\frac{c2^n}{n}}{n^{\ell+k-1}} \frac{\frac{n}{c2^n}}{\left(2^{\frac{n}{c2^n}} - 1\right)}\\[6pt]
        &=\frac1d \lim_{n \to \infty} \frac{\frac{c2^n}{n}}{n^{\ell+k-1}} \frac1{\ln{2}}\\[6pt]
        &=\frac{c}{d \ln{2}} \lim_{n \to \infty} \frac{2^n}{n^{\ell+k}}\\[6pt]
        &= \infty.
\end{align}

This means that $\beta = \infty\not\in\mathbb{R}$, and thus the series diverges in our initial assumption \eqref{eq:setup}, which is a contradiction. Therefore, the value of $\lambda_2$ must be super-polynomial in $n$. 
\end{proof}

\section{Example for artificial criteria enforcement}
\label{app:example}

\subsection{Problem definition}

We show an example of how the criteria enforcement works in theory. The Maximum Cut problem is the usual example for many variational quantum algorithms, as it is a strongly NP-complete problem. The problem is to find a bipartition of nodes in an undirected, unweighted graph that has the most edges between nodes in different parties. The weighted version, where each edge has a non-negative weight, and the objective is to maximize the total weight of such inter-partition edges is also NP-complete in its decision problem version \cite{Karp1972}. As an approximability problem, it is APX-hard, meaning that there is no polynomial-time approximation scheme (PTAS) that is able to approximate the solution arbitrarily, unless P = NP. 

From a computer scientific perspective, to provide context, we would like to note that QIPA with the artificial criteria enforcement not only approximates the optimal solution but does so in fully polynomial time with a bounded error on measurements. Therefore, we show this example here, and later we address any possible issues.

For demonstration purposes, we now consider the problem that we encode into a commonly used quadratic unconstrained binary problem formulation. 

For a graph $G = (V, E)$ and $(i, j) \in E, ~w_{ij} > 0 \in \mathbb{R}$. In this problem $x_i\ (i\in V)$ are binary variables, and their value indicate the partition of the vertices.

A standard form would be
\begin{equation}
    C(x) = \sum_{(i,j) \in E} w_{ij}x_i(1-x_j),
\end{equation}
that provides a simple encoding with $x_i\rightarrow (1-Z_i)/2$ where $Z_i$ is the Pauli Z operator,
\begin{align}
    C(\textbf{Z}) &= \sum_{(i,j) \in E} \frac{w_{ij}}{4} (1-Z_i)(1+Z_j)  \\
    &= -\frac{1}{2}\sum_{i<j, (i,j) \in E} w_{ij} Z_i Z_j+\mathrm{const},
\end{align}
where $\mathrm{const} = \sum_{i<j}w_{ij}/2$. Next, we convert that into an Ising Hamiltonian. Note that we omit weights to individual vertices in our model, that is, setting the external magnetic fields of individual sites to zero in our Hamiltonian. If we look for only the optimal solution and not at its value, we can disregard the constant, which is linear in the number of edges, and disregard the factor in front and arrive at the following minimization problem:

\begin{equation}
    H = \sum_{i<j, (i,j) \in E} w_{ij}Z_iZ_j
\end{equation}

\subsection{Upscaling}

With the artificial scaler, we formulate this as
\begin{equation}
    \hat{H} = \alpha H = \sum_{i<j, (i,j) \in E} \alpha w_{ij} Z_iZ_j
\end{equation}
We would like to note here that the artificial enforcement method here is similar to introducing equally scaled weights.
The problem originally can yield inverse exponentially small gap, both in absolute and relative scale. Thus, for our exponential boost in the generic case, we need an exponentially large $\alpha$.

For the problem Hamiltonian $\hat{H}$, that blows up the eigenvalues, which enlarges the absolute gap from inverse exponential to a linear scale in $n$ while preserving the relative gap's possibly inverse exponential value in the generic case. In other words, for the MaxCut problem, which initially does not satisfy the exponential separation criteria in QIPA$_2$, we defined a simple transformation that enforces the formulation to satisfy the criteria. This also means that the upscaled version of our example generic weighted MaxCut has now become a candidate of the exponential separation. 

\subsection{Numerical example}

We show a numerical example from a MaxCut problem instance, where upscaling enhances convergence. 

We emphasize that this result shows evidence of a practical speedup by the upscaling on an example instance. For generic mathematical boundaries and claims on achievable improvements on any problem in a worst-case scenario, see our error analysis in \ref{sec:error}.

We built a small MaxCut instance with 7 nodes and randomly assigned integer valued weights from 0 (no edge) to 11. The exact graph, the optimal cut and its cost are shown on Figure \ref{fig:maxcutgraph}.

\begin{figure}[ht!]
    \centering
    \includegraphics[width=0.8\textwidth]{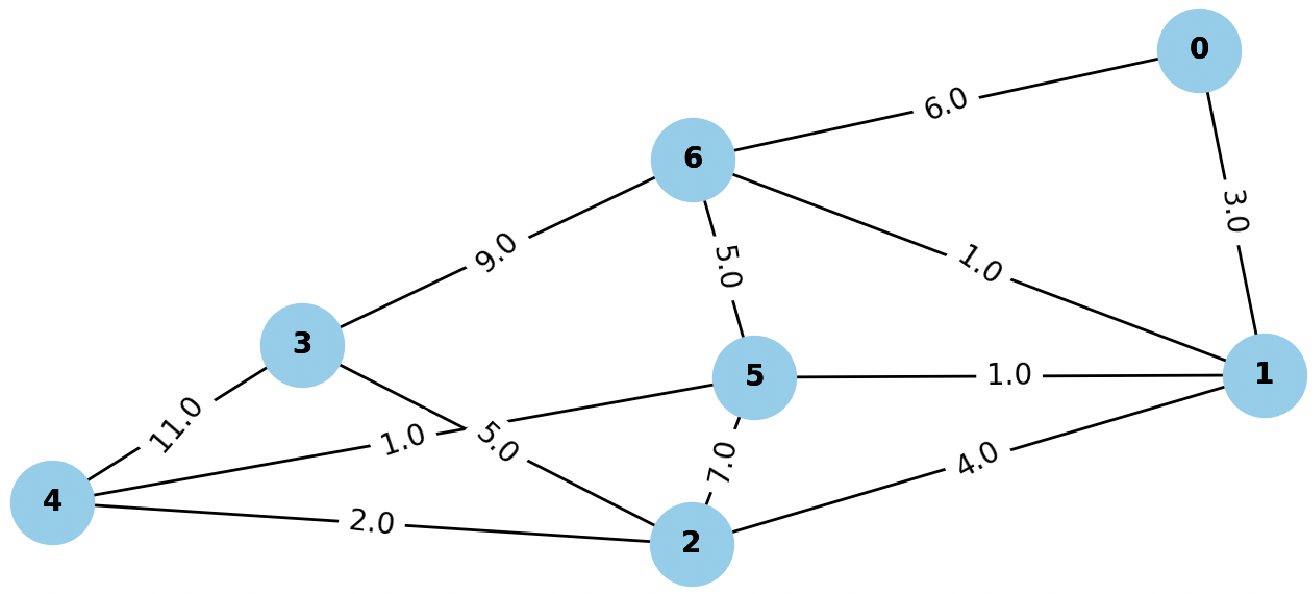}
    \caption{Sample MaxCut instance with 7 nodes, where the optimal partition is \{2,4,6\} and \{0,1,3,5\} with cut value of 49.}
    \label{fig:maxcutgraph}
\end{figure}

We used our own implementation, based on the code available in \cite{Kyaw2023} and integrated it into the Qiskit \cite{qiskit2024} library to run it on IBM's devices. For comparison, we used the varQITE implementation from the Qiskit Algorithms package. The results of the run are shown on Figure \ref{fig:maxcutres}. We can see that the QIPA$_2$ converges faster, getting close to -49.0, which is the true ground state value of the problem Hamiltonian, almost instantly, while varQITE slowly converges. The energy landscape of the example is nice, corresponding to the smooth monotonic curves on the plot. The sections, mostly visible on the QIPA curve, correspond to the iterations of the algorithms.
\begin{figure}[ht!]
    \centering
    \includegraphics[width=0.72\textwidth]{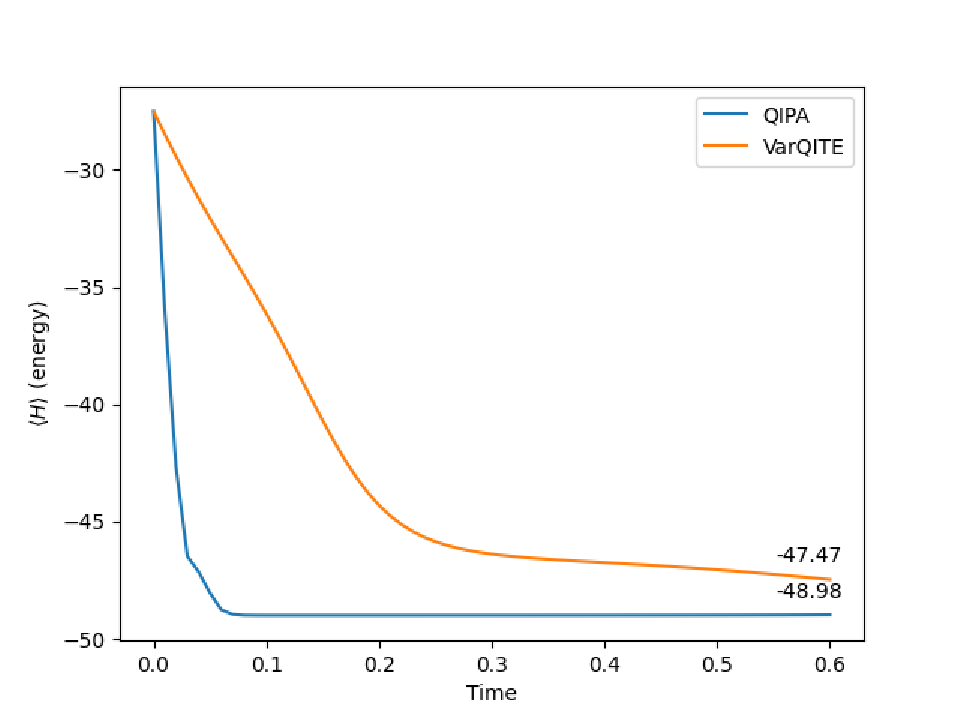}
    \caption{Results of running enhanced, $\alpha = 1.2$ times upscaled QIPA and varQITE on the graph instance above.}
    \label{fig:maxcutres}
\end{figure}

\section{varQITE under upscaling}
\label{app:varQITE}

To see how varQITE behaves under upscaling, we take the derived error bounds from \cite{Zoufal2023} as a basis. It is worth noting that both the algorithmic description and the error measurement in \cite{Zoufal2023} are different from the one we use. However, this error lower bound is also a lower bound for the varQITE algorithm in its QIPA formulation since that phase agnostic formulation only provides tighter bound. The error they provide is Bures distance \cite{Hayashi2006-kf}. 

The Bures metric and the $\ell_2$-norm are equivalent if the time evolution does not introduce change in the global phase or the variational ansatz state manages to perfectly capture the global phase change. In general, it is true that $B \leq \ell_2$, because $B$ does not measure global phase difference.

This means that Bures distance provides a lower bound on the error defined by the $\ell_2$ norm, used for \eqref{eq:error}, and, like trace distance, has quadratic dependence on $H$. 

Their error in Bures distance is defined by a sum of errors from the steps, that in a poly($n$) time algorithm we assume to be a polynomial factor. We work with discrete steps, thus we do not integrate the step size out, which introduces more approximation errors that we do not consider at this point for our bound.
For $|\psi_T^\theta\rangle$ as the parameterized ansatz state for the exact time evolved state $|\psi_T^*\rangle$, it is given by

\begin{equation}
    B(|\psi_T^\theta\rangle, \psi_T^*\rangle) = \delta_t \sum_k^\kappa \norm{|e_{t\delta}\rangle}_2
\end{equation}
where
\begin{equation}
    \norm{|e_t\rangle}_2^2 = Var(H) + \sum_{ij}\Dot{\theta_i}\Dot{\theta_j}F_{ij} + 2\sum_i \Dot{\theta_i}Re(C_i),\label{eq:step_error}
\end{equation}
\begin{equation}
    F_{ij} = Re \left( \frac{\partial \langle \psi_t^\theta |}{\partial \theta_i} \frac{\partial | \psi_t^\theta \rangle}{\partial \theta_j} - \frac{\partial \langle \psi_t^\theta |}{\partial \theta_i} | \psi_t^\theta \rangle \langle \psi_t^\theta | \frac{\partial | \psi_t^\theta \rangle}{\partial \theta_j} \right), \label{eq:F}
\end{equation}

\begin{equation}
    C_i = \frac{\partial \langle \psi_t^\theta |}{\partial \theta_i} H | \psi_t^\theta \rangle. \label{eq:C}
\end{equation}

Since $Var(\alpha H) = \alpha^2 Var(H)$, when we consider $\lambda_1, \lambda_2$ to be the eigenvalues of $H$, the upscaling introduces a quadratic factor to the variance of $H$. That is, when we consider the upscale factor to be $\alpha \sim \exp(n)$, an exponential blow-up in the error occurs. This term appears in the square of the error, but taking square root does not affect the exponential scaling. 

This result shows that for any problem that applies for the exponential separation, regardless of artificially upscaled or not, the error term blows up.

\end{appendices}

\bibliography{qipa-bib}

\end{document}